\documentclass[UKenglish]{lmcs}
\pdfoutput=1

\usepackage{lastpage}
\lmcsdoi{15}{4}{3}
\lmcsheading{}{\pageref{LastPage}}{}{}%
{Nov.~28,~2018}{Oct.~17,~2019}{}

\title[Enumeration of closures and environments]{On the enumeration of closures and environments with an application to random generation}
\author{Maciej Bendkowski}
\address{Jagiellonian University\\
        Faculty of Mathematics and Computer Science\\
        Theoretical Computer Science Department\\
        ul. Prof.~{\L}ojasiewicza 6, 30--348 Krak\'ow, Poland}
\thanks{Maciej Bendkowski was partially supported within the Polish
  National Science Center grant 2016/21/N/ST6/01032.}
        \email{maciej.bendkowski@tcs.uj.edu.pl}

\author{Pierre Lescanne}
\address{University of Lyon\\
        \'Ecole normale sup\'erieure de Lyon\\
        LIP (UMR 5668 CNRS ENS Lyon UCBL)\\
        46 all\'ee d'Italie, 69364 Lyon, France}
        \email{pierre.lescanne@ens-lyon.fr}

\keywords{lambda-calculus,
combinatorics,
functional programming,
mathematical analysis,
complexity}

\usepackage[all]{xy}
\usepackage{qsymbols}
\usepackage{bm}
\usepackage{csquotes}
\usepackage{subfigure}
\usepackage{prooftree}

\usepackage[shortcuts]{extdash}

\usepackage{breqn}

\newcommand{\domsing}[1]{`r_{#1}}

\newcommand{\C}{\mathcal{C}}                            %de Bruijn indices
\newcommand{\E}{\mathcal{E}}                            %de Bruijn indices

\newcommand{\gL}{\mathcal{L}}
\newcommand{\Clos}{\mathcal{C}\mathit{los}}
\newcommand{\Env}{\mathcal{E}\mathit{nv}}

\newqsymbol{"<"}{\langle}
\newqsymbol{">"}{\rangle}
\newqsymbol{`(.)}{\odot}
\newcommand{\zero}{\underline{0}}
\newcommand{\one}{\underline{1}}

\newcommand{\udl}[1]{\underline{#1}}
\newcommand{\nf}{\textsf{nf}}

\definecolor{vertfonce}{rgb}{0,.5,0}
\definecolor{mauve}{rgb}{1,0,1}
\definecolor{rougefonce}{cmyk}{.3,1,.3,0}
\definecolor{cyanp}{cmyk}{.5,.3,0,0}
\definecolor{yellow}{cmyk}{0,0,.7,0}
\definecolor{beige}{cmyk}{0,.2,.7,0}
\definecolor{brun}{cmyk}{0,.5,.7,0}
\definecolor{brunfonce}{cmyk}{.3,.75,.75,.15}

  {\color{blue}}
  {}
  {\color{brunfonce}}
  {}

  {\color{red}}
  {} %

\begin{document}
\begin{abstract}
Environments and closures are two of the main ingredients of evaluation in
    lambda-calculus.  A closure is a pair consisting of a lambda-term and an
    environment, whereas an environment is a list of lambda-terms assigned to
    free variables. In this paper we investigate some dynamic aspects of
    evaluation in lambda-calculus considering the quantitative, combinatorial
    properties of environments and closures. Focusing on two classes of
    environments and closures, namely the so-called plain and closed ones, we
    consider the problem of their asymptotic counting and effective random
    generation. We provide an asymptotic approximation of the number of both
    plain environments and closures of size $n$. Using the associated generating
    functions, we construct effective samplers for both classes of combinatorial
    structures. Finally, we discuss the related problem of asymptotic counting
    and random generation of closed environments and closures.
\end{abstract}
\maketitle

% introduction
\section{Introduction}
\label{sec:introduction}

Though, traditionally, computational complexity is investigated in the context
of Turing machines since their initial development, evaluation complexity in
various term rewriting systems, such as $`l$\=/calculus or combinatory logic,
attracts increasing attention only quite recently. For instance, let us mention
the worst-case analysis of evaluation, based on the invariance of unitary
cost
models~\cite{DBLP:journals/corr/abs-1208-0515,DBLP:conf/rta/AvanziniM10,DBLP:journals/corr/AccattoliL16}
or transformation techniques proving termination of term rewriting
systems~\cite{DBLP:conf/icfp/AvanziniLM15}.

Much like in classic computational complexity, the corresponding average-case
analysis of evaluation in term rewriting systems follows a different, more
combinatorial and quantitative approach, compared to its worst-case variant.
In~\cite{DBLP:conf/rta/ChoppyKS87,DBLP:journals/tcs/ChoppyKS89} Choppy, Kaplan
and Soria propose an average-case complexity analysis of normalisation in a
general class of term rewriting systems using generating functions, in
particular techniques from analytic combinatorics~\cite{flajolet09}. Following a
somewhat similar path, Bendkwoski, Grygiel and Zaionc investigated later the
asymptotic properties of normal-order reduction in combinatory logic, in
particular the normalisation cost of large random
combinators~\cite{bengryzai2017,BENDKOWSKI_2017}. Alas, normalisation in
$`l$\=/calculus has not yet been studied in such a combinatorial context.
Nonetheless, static, quantitative properties of $`l$\=/terms, form an active
stream of recent research. Let us mention, non-exhaustively, investigations into
the asymptotic properties of large random
$`l$\=/terms~\cite{dgkrtz,BendkowskiGLZ16} or their effective counting and
random generation ensuring a uniform distribution among terms with equal
size~\cite{DBLP:conf/analco/BodiniGG11,gryles2015,GittenbergerGolebiewskiG16,BodiniGitGol17}.

In the current paper, we take a step towards the average-case analysis of
reduction complexity in $`l$\=/calculus. Specifically, we offer a quantitative
analysis of environments and closures --- two types of structures frequently
present at the core of abstract machines modelling $`l$\=/term evaluation, such
as for instance the Krivine or U- machine~\cite{Curien:1994,Lescanne1994},
presented in Section~\ref{sec:machine}.
In Section~\ref{sec:closures-and-environments} we discuss the combinatorial
representation of environments and closures, in particular the associated
de~Bruijn notation. In Section~\ref{sec:analytic-tools} we list the analytic
combinatorics tools required for our analysis and we show in
Section~\ref{sec:random} how they can be used for random generation.
In Section~\ref{sec:counting-plain} and Section~\ref{sec:counting-closed} we conduct our
quantitative investigation into so-called plain and closed environments and
closures, respectively, subsequently concluding the paper
in Section~\ref{sec:conclusions}.

\section{A combinatoric approach to higher order rewriting systems}
\label{sec:motiv}

As said in the introduction, viewing the $`l$-calculus from the
perspective of counting is new, especially in the scientific community of
structures for computation and deduction and requires motivation to be
detailed.

First, clearly a new perspective on $`l$-calculus enlightens the semantics
and opens new directions, especially by adding a touch of efficiency and a
discussion on how the size of structures with binders (like  $`l$-terms)
can be measured.  However, despite  advanced mathematical techniques are used, the goal is
more practical and connected to operational semantics and
implementation. Counting allows assigning a precise measure on how a
specific algorithm performs.  In
~\cite{knuth00:math_Anal_Algo}\footnote{This paper is part of the
  book ``Selected Papers on Analysis of Algorithms''~\cite{knuth00:_selec_paper_analy_algor} dedicated to Professor
  N.~G. de~Bruijn.}  Knuth calls analysis of Type~A an \emph{analysis of a
  particular algorithm} and shows how important it is in computer
science. He adds (p.~3): ``Complexity analysis provides an interesting way
to sharpen our tools for the more routine problems we face from day to
day.''

Furthermore, a notion of probabilistic distribution as used in the
average-case analysis of algorithms, after Sedgewick and
Flajolet~\cite{sedgewick2014introduction}, is deduced.  In particular a
notion of uniform distribution is inferred in order to evaluate the
\emph{average case efficiency} of algorithms w.r.t. this distribution.  In
this paper, the algorithms the authors have in mind are the several
reduction machines for the $`l$-calculus, especially the Krivine machine
and the U-machine, for which analyses of Type~A and more specifically
average case analyses are expected to be built.  Another application is
\emph{random generation} of terms and several kinds of logical models for
computation as used for instance in
QuickCheck~\cite{Claessen-2000}.  A fully and mathematically justified
random generator can only be built using the kind of tools developed in
this paper.

But average case analysis based on uniform distribution is not the only
one. The so-called \emph{smoothed analysis of
  algorithms}~\cite{DBLP:journals/jacm/SpielmanT04} is another family of
tools which is based on measures of size. Here the distribution is no more
uniform and this method has promising applications, hopefully in structures for
computation.

% closures and environments
\section{Environments and closures}
\label{sec:closures-and-environments}
In this section we outline the de~Bruijn notation and related concepts deriving
from $`l$\=/calculus variants with explicit substitutions used in the
subsequent sections.

\subsection{De~Bruijn notation}
Though the classic variable notation for $`l$\=/terms is elegant and concise, it
poses considerable implementation issues, especially in the context of
substitution resolution and potential name clashes. In order to accommodate
these problems, de~Bruijn proposed an alternative name-free notation for
$`l$\=/terms~\cite{deBruijn1972}.  In this notation, each variable $x$ is
replaced by an appropriate non-negative integer $\underline{n}$ (so-called
\emph{index}) intended to encode the distance between $x$ and its binding
abstraction. Specifically, if $x$ is bound to the $(n+1)$st abstraction on its
unique path to the term root in the associated $`l$\=/tree, then $x$ is replaced
by the index $\underline{n}$. In this manner, each closed $`l$\=/term in the
classic variable notation is representable in the de~Bruijn notation.

\begin{exa}\label{exa:ThreeTrees}
Consider the $`l$\=/term $T = (`l x y z u. x (`l y x .y))~(`l z . (`l u
.  u) z)$.~\autoref{fig:threeTerms} depicts three different representations of
$T$ as tree-like structures. The first one uses explicit variables, the second
one uses back pointers to represent the bound variables, whereas the third one
uses De Bruijn indices.

\begin{figure}[ht!]
  \centering
  \begin{displaymath}
  \resizebox{0.8\textwidth}{!}{
    \xymatrix @C=7pt@R=7pt{&&@\ar@{-}[dl]\ar@{-}[dr]\\ % 1
      &{`l x}\ar@{-}[d]&&`l z\ar@{-}[d]\\ % 2
      &`l y \ar@{-}[d]&&@\ar@{-}[dl]\ar@{-}[dr]\\ % 3
      &`l z \ar@{-}[d]&`l u\ar@{-}[d]&& z\\  % 4
      &`l u\ar@{-}[d]& u\\ %5
      &@\ar@{-}[dl]\ar@{-}[dr]\\ %6
      x && `l y\ar@{-}[d]\\ %7
      &&`l x\ar@{-}[d]\\ % 8
      && y
    }
    \qquad \qquad \qquad
    \xymatrix @C=7pt@R=7pt{&&@\ar@{-}[dl]\ar@{-}[dr]\\ % 1
      &{`l}\ar@{-}[d]&&`l\ar@{-}[d]\\ % 2
      &`l\ar@{-}[d]&&@\ar@{-}[dl]\ar@{-}[dr]\\ % 3
      &`l\ar@{-}[d]&`l\ar@{-}[d]&&\ar@[back]@(dr,r)[uul]\\  % 4
      &`l\ar@{-}[d]&\ar@[black,dashed]@(d,r)[u]\\ %5
      &@\ar@{-}[dl]\ar@{-}[dr]\\ %6
      \ar@[black]@(dl,l)[uuuuur] && `l\ar@{-}[d]\\ %7
      &&`l\ar@{-}[d]\\ % 8
      &&\ar@[black]@(d,r)[uu]
    }
    \qquad \qquad \qquad
    \xymatrix @C=7pt@R=7pt{&&@\ar@{-}[dl]\ar@{-}[dr]\\ % 1
      &{`l}\ar@{-}[d]&&`l\ar@{-}[d]\\ % 2
      &`l\ar@{-}[d]&&@\ar@{-}[dl]\ar@{-}[dr]\\ % 3
      &`l\ar@{-}[d]&`l\ar@{-}[d]&&\underline{0}\\  % 4
      &`l\ar@{-}[d]&\underline{0}\\ %5
      &@\ar@{-}[dl]\ar@{-}[dr]\\ %6
      \underline{3} && `l\ar@{-}[d]\\ %7
      &&`l\ar@{-}[d]\\ % 8
      &&\underline{1}
    }}
  \end{displaymath}
  \caption{Three representations of the $`l$\=/term $T = (`l x y z u. x (`l y
    x .y))~(`l z . (`l u . u) z)$.}
  \label{fig:threeTerms}
\end{figure}
In order to represent free occurrences of variables, one uses indices of values
exceeding the number of abstractions crossed on respective paths to the term
root. For instance, $`lx. y z$ can be represented as $`l \underline{1}
\underline{2}$ since $\underline{1}$ and $\underline{2}$ correspond to two
different variable occurrences.
\end{exa}

Recall that in the classic variable notation a~$`l$\=/term $M$ is said to be
\emph{closed} if each of its variables is bound.  In the de~Bruijn notation, it
means that for each index occurrence $\underline{n}$ in $M$ one finds at least
$n+1$ abstractions on the unique path from $\underline{n}$ to the term root of
$M$.  If a $`l$\=/term is not closed, it is said to be \emph{open}.  If
heading $M$ with $m$ abstractions turns it into a closed $`l$\=/term, then
$M$ is said to be \emph{$m$-open}. In particular, closed $`l$\=/terms are $0$-open.

\begin{exa}
Note that $`l `l `l `l (\underline{3}(`l `l \underline{1}))~(`l (`l
    \underline{0})\underline{0})$, actually the $T$ of
    Example~\ref{exa:ThreeTrees} in De Bruijn notation, is closed. The $`l$-term
    $\underline{3}(`l `l \underline{1})$ is $4$-open, however it is not
    $3$-open.  Indeed, $`l `l `l (\underline{3}(`l `l \underline{1}))$ is
    $1$-open instead of being closed.  Similarly, $`l (\underline{3}(`l `l
    \underline{1}))$ is $3$-open, however it is not $2$-open.
\end{exa}

\begin{exa}
  Consider, on Figure~\ref{fig:SK}, the term \textsf{S\,K} and its two direct contractions
  \begin{displaymath}
\newcommand{\ul}[1]{\underline{#1}}
    (`l`l`l\,\ul{2}\,\ul{0}\,(\ul{1}\,\ul{0}))\,(`l`l\ul{1}) "->" %
    `l`l((`l`l\ul{1})\,\ul{0}\,(\ul{1}\,\ul{0})) "->"%
    `l`l((`l\ul{1})\,(\ul{1}\,\ul{0}),
  \end{displaymath}
or, in notation with explicit names
\begin{displaymath}
  (`l x .`l y. `l z . x z (y z))\,(`l x. `l y. x) "->" %
`l y .`l z. (`l x . `l y . x) z (y z)  "->"%
`l y . `l z . (`l y . z) (y z).
\end{displaymath}
It shows how $`b$-contraction works in De Bruijn notation (cf.~the next
    subsection).  Moreover, it shows in
    $`l`l(`l\underline{1}\,(`l\underline{1}\,\underline{0})$ that the same
    variable namely $z$ may be associated with two De Bruijn indices, namely
    $\underline{1}$ and $\underline{0}$ and that the same De Bruijn index namely
    $\underline{1}$ may be associated with two variables namely $y$ and $z$.  In
    the de~Bruijn notation the value of an index associated with a variable
    depends of the context.
\end{exa}

\begin{figure}[th!]
  \centering
 \resizebox{0.9\textwidth}{!}{
  \xymatrix @C=7pt@R=7pt{
&&&&{@}\ar@{-}[dl]\ar@{-}[dr] \\
&&&`l\ar@{-}[d]&&`l\ar@{-}[d]&\\
&&&`l\ar@{-}[d]&&`l\ar@{-}[d]&\\
&&&`l\ar@{-}[d]&&\ar@[black]@(d,r)[uu]\\
&&&{@}\ar@{-}[dll]\ar@{-}[drr]\\
&{@}\ar@{-}[dl]\ar@{-}[dr]&&&&{@}\ar@{-}[dl]\ar@{-}[dr]& \\
\ar@(dl,l)[uuuuurrr] &&\ar@(dr,l)[uuur]&&\ar@(dl,r)[uuuul] &&\ar@(dr,r)[uuulll]  &
}
\quad\raisebox{-.1\textwidth}{$"->"$}
  \xymatrix @C=7pt@R=7pt{
&&&`l\ar@{-}[d]\\
&&&`l\ar@{-}[d]\\
&&&{@}\ar@{-}[dll]\ar@{-}[drr]\\
&@\ar@{-}[dl]\ar@{-}[dr]&&&&@\ar@{-}[dl]\ar@{-}[dr]\\
`l\ar@{-}[d]&&\ar@(dr,l)[uuur]&&\ar@(dl,r)[uuuul] &&\ar@(dr,r)[uuulll]\\
`l\ar@{-}[d]\\
\ar@(d,r)[uu]&
}
\quad\raisebox{-.1\textwidth}{$"->"$}\quad
  \xymatrix @C=7pt@R=7pt{
&&`l\ar@{-}[d]\\
&&`l\ar@{-}[d]\\
&&@\ar@{-}[dll]\ar@{-}[drr] \\
`l\ar@{-}[d]&&&&@\ar@{-}[dl]\ar@{-}[dr]\\
\ar@(dl,l)[uuurr] &&&\ar@(dl,r)[uuuul]&&\ar@(dr,r)[uuulll]
}
}
\caption{The term \textsf{S} \textsf{K} and two contractions.}\label{fig:SK}
\end{figure}

Certainly, the set $\gL_m$ of $m$-open terms is a subset of the set of
$(m+1)$-open terms. In other words, if $M$ is $m$-open, it is also $(m+1)$-open.
The set of all $`l$\=/terms is called the set of \emph{plain} terms.  It is the
union of the sets of $m$-open terms and is denoted as $\gL_\infty$. Hence,
\begin{equation}
  \gL_0 ``(= \gL_1 ``(= \cdots ``(= \gL_m ``(= \gL_{m+1} \cdots ``(=
    \bigcup_{i=0}^\infty \gL_i = \gL_\infty\, .
\end{equation}

Let us note that de~Bruijn's name-free representation of $`l$\=/terms exhibits
an important combinatorial benefit. Specifically, each $`l$\=/term in the
de~Bruijn notation represents an entire $`a$-equivalence class of $`l$\=/terms
in the classical variable notation. Indeed, two variable occurrences bound by
the same abstraction are assigned the same de~Bruijn index. In consequence,
counting $`l$\=/terms in the de~Bruijn notation we are, in fact, counting entire
$`a$-equivalence classes instead of their inhabitants.

\subsection{Closures and $`b$-reduction}
Recall that the main rewriting rule of $`l$-calculus is \emph{$`b$-reduction},
see, e.g.~\cite{Curien1996}:
\begin{equation}
  \begin{array}{l rcl}
    (`b)& (`l M)\ N &"->"& M \{\underline{0} "<-" N\}
  \end{array}
\end{equation}
where the operation $\{\underline{n} "<-" M\}$, i.e.~substitution of $`l$\=/terms
for de~Bruijn indices, is defined inductively as follows:
\begin{align}
\begin{split}
  (M~N)\{\underline{n} "<-" P\} &= M\{\underline{n} "<-"
  P\}~N\{\underline{n} "<-" P\}\\
(`l M)\{\underline{n}"<-" P\} &=`l (M\{\underline{(n+1)}"<-" P\})\\
\underline{m} \{\underline{n}"<-" P\}&=
    \begin{cases}
    \underline{m-1} & \textrm{if~}m > n\\
\tau^n_0(P) & \textrm{if~} m = n\\
\underline{m} & \textrm{if~} m < n\, .
    \end{cases}
\end{split}
\end{align}
The first rule distributes the substitution in an application, the second rule
pushes a substitution under an abstraction, and the third rule dictates how a
substitution acts when the term is an index.  Finally, $\tau^n_0(P)$ tells how
to update the indices of a term which is substituted for an index. The operation
$\tau^n_i(M)$ is defined by induction on $M$ as follows:
\begin{align}
\begin{split}
  \tau^n_i(M~N) &=  \tau^n_i(M)~\tau^n_i(N)\\
\tau^n_i(`l M) &=`l(\tau^n_{i+1}(M))\\
\tau^n_i(\underline{m}) &=
    \begin{cases}
        \underline{m+n-1} & \textrm{if~} m > i\\
        \underline{m} & \textrm{if~} m \le i\, .
    \end{cases}
\end{split}
\end{align}

A $`l$\=/term in the form of $(`l M)\ N$ is called a \emph{$`b$-redex} (or
simply a \emph{redex}). Lambda terms not containing $`b$-redexes as subterms,
are called ($`b$-)normal forms. The computational process of rewriting
(reducing) a $`l$\=/term to its $`b$\=/normal form by successive
elimination of $`b$\=/redexes is called \emph{normalisation}. There exists an
abundant literature on normalisation in $`l$\=/calculus; let us mention, not
exhaustively~\cite{landin64,Plotkin75,DBLP:conf/lfp/MaunyS86,Curien:1994,Mitchell2002}.

The central concepts present of formalisms dealing with
normalisation in $`l$\=/calculus are environments and closures. An
\emph{environment} is a list of not yet evaluated closed terms meant to be assigned to indices
$\underline{0}, \underline{1},\underline{2}, \ldots, \underline{m-1}$ of an
$m$-open $`l$\=/term.  As lists, environments have two basic operations
(two basic constructors), namely $\Box$ for the empty environment and ``$:$'' for
the \emph{cons} operator, i.e., for the operator that put an item in front of an environment.
Those not fully evaluated closed terms are represented by
\emph{closures}, where
a \emph{closure} is a couple
consisting of an $m$\=/open $`l$\=/term and an environment.  For instance, the
closure $"<"M,\Box">"$ consists of the $`l$\=/term $M$ evaluated in the context
of an empty environment, denoted as~$\Box$, and represents simply $M$. The
closure $"<"\zero\, \one, "<"`l`l\zero,\Box">":"<"`l\zero,\Box">":\Box">"$
represents the $`l$\=/term $(\zero\, \one)$ evaluated in the context of an
environment $"<"`l`l\zero,\Box">":"<"`l\zero,\Box">":\Box$. Here, intuitively,
the index $\zero$ receives the value $`l`l\zero$ whereas the index $\one$ is
assigned to $`l\zero$. Finally, $`l`l\zero$ is applied to $`l\zero$. And
so, reducing the closure $"<"\zero\, \one,
"<"`l`l\zero,\Box">":"<"`l\zero,\Box">":\Box">"$, for instance using a
Krivine abstract machine~\cite{Curien:1994} (see Section~\ref{sec:Kmachine}), we obtain $`l\zero$.

Let us notice that following the outlined description of environments and
closures, we can provide a formal combinatorial specification for both using the
following mutually recursive definitions:
\begin{align}\label{eq:closure:env:def}
\begin{split}
  \Clos &::= "<" `L, \Env ">" \\
  \Env &::= \Box \mid \Clos : \Env
\end{split}
\end{align}
In the above specification, $`L$ denotes the set of all plain
$`l$\=/terms.  Moreover, we introduce two binary operators ``$"<"\_,
\_">"$'', i.e.~the \emph{coupling} operator, and ``$:$'', i.e.~the
\emph{cons} operator, heading its left-hand side on the right-hand
list.  When applied to a $`l$\=/term and an environment, the coupling
operator constructs a new closure. In other words, a \emph{closure} is
a couple of a $`l$\=/term and an environment whereas an environment is
a list of closures, representing a list of assignments to free
occurrences of de~Bruijn indices.

Such a combinatorial specification for closures and environments plays an
important r\^ole as it allows us to investigate, using methods of analytic
combinatorics, the quantitative properties of both closures and environments.

\section{Closures and abstract machines}
\label{sec:machine}

Closures are one of the main ingredients of abstract machines performing
reduction in $\lambda$\nobreakdash-calculus. In the current section, we briefly
mention two such machines and discuss how closures and environments relate to
the evaluation dynamics of $`l$\nobreakdash-terms.

\subsection{The Krivine machine}
\label{sec:Kmachine}

The presentation of the Krivine machine we give here can be found in Curien's
book~\cite[p. 66]{Curien:1994}. The state of the machine is a non-empty
environment. Its transitions are:
\begin{displaymath}
\begin{array}{lcl@{\qquad}l}
  "<" M\,N, e">" : e' &"->"& "<" M, e ">" : "<"N, e">" : e' & (App)\\
  "<"`l M, e ">" :  "<"N, f">" : e' &"->"& "<"M, "<"N, f">" : e">" : e'&(Abs)\\
  "<"\zero,"<"M, f">":e">":e' &"->"& "<" M, f">":e'&(Zero)\\
  "<"\udl{n+1},"<"M, f">":e">":e' &"->"& "<" \udl{n},e">":e'& (Succ)
\end{array}
\end{displaymath}

Interestingly, it is possible to optimise the above transition rules by
merging the rules $(Zero)$ and $(Succ)$ into a single rule $(Fetch)$ given as
\begin{displaymath}
\begin{array}{lcl@{\qquad}l}
  "<"\udl{i},"<"M_0, f_0">":\ldots:"<"M_i, f_i">":\ldots">":e' &"->"& "<"
  M_i, f_i">":e'& (Fetch)
\end{array}
\end{displaymath}
In words, when interpreting the index $\udl{i}$ we evaluate the $i$th closure of
the environment associated with this index. Consequently, a sequence of $i+1$
transitions is replaced by a single one.  The above Krivine machine performs
head reductions and hence implements a \emph{call-by-name} evaluation strategy.
Strong normalisation can be implemented using, e.g.~the U\nobreakdash-machine.

\subsection{The U-machine}
\label{sec:Umachine}

The U-machine is an abstract machine derived from the calculus of explicit
substitution $`l`y$,
see~\cite{Lescanne1994,DBLP:conf/compass/Lescanne95,benaissa_briaud_lescanne_rouyer-degli_1996}.
First, let us recall that a term of the $`l`y$\nobreakdash-calculus can contain
\emph{explicit substitutions} in form of $M[s]$ where $M$ is a \emph{term} and
$s$ a \emph{substitution} as in the following grammar:
\begin{eqnarray*}
  M, N &::=& M\,N \mid `l M \mid \udl{n} \mid M[s]\\
  s &::=& M/ ~\mid~ \Uparrow(s) ~\mid~ \uparrow.
\end{eqnarray*}
New operators corresponding to components of explicit substitutions admit the
following, intuitive meaning.  The \emph{slash} operator $/$ turns a given term
into a substitution. Intuitively, it is meant to assign the given term $M$ to
the index $\zero$ as in $\zero[M/] \to M$. The \emph{shift} operator  $\uparrow$
is a constant whose role is to increment de~Bruijn indices, for instance
$\underline{n}[\uparrow] \to \underline{n+1}$.  Finally \emph{lift}, denoted as
$\Uparrow$, is meant to adjust the explicit substitution in the case when it is
pushed under an abstraction. For instance, $(`l N)[s] \to `l(N[\Uparrow(s)])$.

Formally, the way $`b$-reduction and explicit substitutions work together is
given by the following rules of the $`l`y$-calculus:
\begin{displaymath}
\begin{array}{lcl@{\qquad}l}
(`l M)\,N&"->"& M[N/] &(Beta)\\
M\,N[s]&"->"& M[s]\,N[s] & (App)\\
(`l M)[s] &"->"& `l(M[\Uparrow(s)])&(Abs)\\
\udl{0}[M/] &"->"& M &(FVar)\\
\udl{n+1}[M/]&"->"& \udl{n}& (RvAr)\\
\udl{0}[\Uparrow(s)]&"->"& \udl{0}&(FVarLift)\\
\udl{n+1}[\Uparrow(s)]&"->"& \udl{n}[s][\uparrow]&(RVarLift)\\
\udl{n}[\uparrow]&"->"& \udl{n+1} &(VarShift)
\end{array}
\end{displaymath}

In the U-machine \emph{environments} are modified so to fit with the features of
the $`l`y$\nobreakdash-calculus, especially with the shift and lift operators.
Environments are still lists of \emph{operations} to be performed on variables.
These operations, in turn, are pairs in form of $(a,i)$ where $i$ is the number
of lifts to be executed before basic actions are performed. Finally, basic
actions are of two forms; either they are a \emph{shift} $\uparrow$, or a closure $"<"M,e">"$.
In other words, closures and environments of the U-machine are changed into:
\begin{displaymath}
\begin{array}{lcl@{\qquad}l}
  e, f, g &::=& (a,i)^* &(\textrm{lists of operations})\\
  a &::=& \uparrow ~\mid~ "<" M, e">"& \textrm{(basic actions})\\
  i &::=& 0 ~\mid~ i+1 & (\textrm{number of lifts})
\end{array}
\end{displaymath}

A state of the U-machine is a list $"<"M,e">"^*$ of pairs where $M$ is a
\emph{term} and $e$ is a \emph{list of operations}.  Let $(++)$ denote list
concatenation and \textsf{LiftEnv} denote the map incrementing all the second
arguments of given list of pairs, i.e.~a coordinate-wise function $(a,i) \mapsto
(a,i+1)$. Then, the transitions of the U-machine are given as follows:
\begin{displaymath}
\begin{array}{lcl@{\qquad}l}
  "<" M\,N, f">" : e &"->"& "<" M, f ">" : "<"N, f">" : e & (APP)\\
  "<"`l M, f">" :  "<"N, g">" : e &"->"& "<"M, \mathsf{LiftEnv}(f) ++ ["<"N,g">"]">" : e&(LBA-BET)\\
  "<"\zero,(a,i+1):f">":e &"->"& "<" \zero, f">":e&(FVARLIFT)\\
  "<"\udl{n+1},(a,i+1):f">":e &"->"& "<" \udl{n},
  (a,i):(\uparrow,0):f">":e& (RVARLIFT)\\
"<" \udl{0}, ("<"M,f">",0):g">":e&"->"& "<"M, f ++ g">":e&(FVAR)\\
"<"\udl{n+1},  ("<"M,f">",0):g">":e&"->"& "<"\udl{n},g">":e&(RVAR)\\
"<"\udl{n}, (\uparrow, 0):g">":e&"->"& "<"\udl{n+1}, g">": e&(VARSHIFT)
\end{array}
\end{displaymath}
\newcommand{\aU}{\raisebox{-5pt}{$\stackrel{\xymatrix{\ar@{->|}[r]&}}{\scriptscriptstyle
      U}$}} %
\newcommand{\anf}{\mathop{\downarrow \!\!{\scriptstyle\nf}}} %

In the U-machine, two kind of states cannot be further reduced, i.e.~states of
the form $"<"`l N, f">": \Box$ (abstractions with empty stacks) and states of
the form $"<"\udl{n},\Box">":f$ (indices with nothing in their direct
environment). It is possible to further reduce those states using strong
normalisation.  For that, we introduce the following inference rules which correspond to
recursive calls of the U-machine. In there inference rules, $\!\!\aU$ is a
relation between list of pairs in form of $"<"M,e">"$ and corresponds to the
reduction to normal form. Moreover, $\anf$ is a deterministic relation between a
closure and a term.  When we want to designate the result $N$ of the relation
$\anf$ we write $\nf"<"M, e">"$ instead of $"<"M, e">" ~ \anf ~ N$.
\begin{displaymath}
  \prooftree
"<"M, e">":\Box  \aU "<"`l N, f">":\Box
    \justifies "<"M,e">" \quad \anf \quad `l (\nf"<"N, \textsf{LiftEnv}(f)">")
  \endprooftree
\end{displaymath}
\bigskip
\begin{displaymath}
  \prooftree
"<"M, e">":\Box  \aU "<"\udl{n}, \Box">":f
\justifies "<"M, e">"  \quad \anf \quad  \udl{n}~ (\mathsf{map}~ \nf~f)
  \endprooftree
\end{displaymath}

%\medskip
Actually, $\udl{n}~ (\mathsf{map}~ \nf~f)$ is an abuse of notation for the
successive applications of the list $\mathsf{map}~ \nf~f$ on $\udl{n}$.

% analytic tools
\section{Analytic tools}
\label{sec:analytic-tools}
In the following section we briefly\footnote{In such a short presentation
  of a non-trivial theory, many terms, like ``branch'', ``Newton-Puiseux
  series'', ``locally convergent'' etc.  are not defined. They are defined
  in the references~\cite{flajolet09,Wilf2006,ghys2017singular}.} outline
the main techniques and notions from the theory of generating functions
and singularity analysis. We refer the curious reader
to~\cite{flajolet09,Wilf2006,ghys2017singular} for a thorough
introduction.

Let $\left(f_n\right)_n$ be a sequence of non-negative integers. Then, the
\emph{generating function} $F(z)$ associated with $\left(f_n\right)_n$ is the
formal power series  $F(z) = \sum_{n\geq 0} f_n z^n$. Following standard
notational conventions, we use $[z^n]F(z)$ to denote the coefficient standing by
$z^n$ in the power series expansion of $F(z)$. Given two sequences
$\left(a_n\right)_n$ and $\left(b_n\right)_n$ we write $a_n \sim b_n$ to denote
the fact that both sequences admit the same asymptotic growth order,
specifically $\displaystyle\lim_{n \to \infty} \dfrac{a_n}{b_n} = 1$. Finally,
we write $\varphi \doteq c$ when the
expression $\varphi$ is approximated by the number~$c$.

Suppose that $F(z)$, viewed as a function of a single complex variable $z$, is
defined in some region $\Omega$ of the complex plane centred at $z_0 \in
\Omega$. Then, if $F(z)$ admits a convergent power series expansion in form of
\begin{equation} F(z) = \sum_{n \geq 0} f_n {(z - z_0)}^n \end{equation} it is
    said to be \emph{analytic} at point $z_0$. Moreover, if $F(z)$ is analytic
    at each point $z \in \Omega$, then $F(z)$ is said to be \emph{analytic in the
    region $\Omega$}. Suppose that there exists a function $G(z)$ analytic in a
    region $\Omega^*$ such that $\Omega \cap \Omega^* \neq \emptyset$ and both
    $F(z)$ and $G(z)$ agree on $\Omega \cap \Omega^*$, i.e.~$F\rvert_{\Omega
    \cap \Omega^*} = G\rvert_{\Omega \cap \Omega^*}$, where $F\rvert_A$ is
  the restriction of the function $F$ on the region $A$. Then, $G(z)$ is said to be
    an \emph{analytic continuation} of $F(z)$ onto $\Omega^*$. If $F(z)$ defined
    in some region $\Omega \setminus \{z_0\}$ has no analytic continuation onto
    $\Omega$, then $z_0$ is said to be a \emph{singularity} of $F(z)$.
When a formal power series $F(z) = \sum_{n\geq 0} f_n z^n$ represents an
analytic function in some neighbourhood of the complex plane origin, it becomes
possible to link the location and type of singularities corresponding to $F(z)$,
in particular so-called \emph{dominating} singularities residing at the
respective circle of convergence, with the asymptotic growth rate of its
coefficients.  This process of \emph{singularity analysis} developed by Flajolet
and Odlyzko~\cite{FlajoletOdlyzko1990} provides a general and systematic
technique for establishing the quantitative aspects of a broad class of
combinatorial structures.

While investigating environments and closures, a particular example of algebraic
combinatorial structures, the respective generating functions turn out to be
algebraic themselves. The following prominent tools provide the essential
foundation underlying the process of \emph{algebraic singularity analysis} based
on \emph{Newton-Puiseux expansions}, i.e.~extensions of power series allowing
fractional exponents.

\begin{thmC}[Newton,
    Puiseux~{\cite[Theorem~VII.7]{flajolet09}}]\label{th:newton-puiseux}
    Let $F(z)$ be a branch of an algebraic equation $P(z, F(z)) = 0$. Then, in a
    circular neighbourhood of a singularity $\rho$ slit along a ray emanating
    from $\rho$, $F(z)$ admits a fractional Newton-Puiseux series expansion that
    is locally convergent and of the form \begin{equation} F(z) = \sum_{k \geq
    k_0} c_k {\left( z - \rho \right)}^{k/\kappa} \end{equation}
    where $k_0 \in \mathbb{Z}$ and $\kappa \geq 1$.
\end{thmC}
Let $F(z)$ be analytic at the origin. Note that $[z^n]F(z) =
\rho^{-n}[z^n]F(\rho z)$. In consequence, following a proper rescaling we can
focus on the type of singularities of $F(z)$ on the unit circle. The standard
function scale provides then the asymptotic expansion of $[z^n]F(z)$.

\begin{thmC}[Standard function
    scale~{\cite[Theorem~VI.1]{flajolet09}}]\label{th:standard-func-scale}
      Let $\alpha \in \mathbb{C} \setminus \mathbb{Z}_{\leq 0}$. Then, $F(z) =
      {(1 - z)}^{-\alpha}$ admits for large $n$ a complete asymptotic expansion
      in form of \begin{equation} [z^n]F(z) =
          \frac{n^{\alpha-1}}{\Gamma(\alpha)} \left( 1 +
          \frac{\alpha(\alpha-1)}{2n} +
      \frac{\alpha(\alpha-1)(\alpha-2)(3\alpha-1)}{24n^2} +
      O\left(\frac{1}{n^3}\right) \right) \end{equation} where $\Gamma \colon
      \mathbb{C} \setminus \mathbb{Z}_{\leq 0} \to \mathbb{C}$ is the Euler
      Gamma function defined as \begin{equation} \Gamma(z) = \int_{0}^{\infty}
      x^{z-1} e^{-x} dx \qquad \textrm{for~} \Re(z) > 0  \end{equation}
and by analytic continuation on all its domain.
\end{thmC}
Given an analytic generating function $F(z)$ implicitly defined as a branch of
an algebraic function satisfying $P(z, F(z)) = 0$, our task of establishing the
asymptotic expansion of the corresponding sequence ${\left([z^n]F(z)\right)}_n$
reduces to locating and studying the (dominating) singularities of $F(z)$.  For
generating functions analytic at the complex plane origin, this quest simplifies
even further due to the following classic result.

\begin{thmC}[Pringsheim~{\cite[Theorem~IV.6]{flajolet09}}]\label{th:pringsheim}
  If $F(z)$ is representable at the origin by a series expansion that has
    non-negative coefficients and radius of convergence $R$, then the point $z =
    R$ is a singularity of $F(z)$.
\end{thmC}
We can therefore focus on the real line while searching for respective
singularities. Since $\sqrt{z}$ cannot be unambiguously defined as an analytic
function at $z=0$ we primarily focus on roots of radicand expressions in the
closed-form formulae of investigated generating functions.

\subsection*{Counting $`l$-terms}
Let us outline the main quantitative results concerning $`l$\=/terms in the
de~Bruijn notation, see~\cite{Bendkowski2016,BendkowskiGLZ16,GittenbergerGolebiewskiG16}. In
this combinatorial model, indices are represented in a unary encoding
using the successor operator $\mathsf{S}$ and $0$. In the so-called
\emph{natural} size notion~\cite{BendkowskiGLZ16}, assumed throughout the
current paper, the size of $`l$\=/terms is defined recursively as follows:
\begin{displaymath}
  \begin{array}{l@{\hspace*{50pt}}l}
    \begin{array}{lcl}
      | 0 | &=& 1\\
      | \mathsf{S}~ n| &=& |\underline{n}| ~~ = ~~ |n| + 1
    \end{array}
    &
    \begin{array}{lcl}
      |M \, N| &=& |M| + |N| + 1\\
      |`l M | &=& |M| + 1\, .
    \end{array}
  \end{array}
\end{displaymath}
And so, for example, $|`l \underline{1} \underline{2}| = 7$.

\begin{rem}
We briefly remark that different size notions in the de~Bruijn representation,
    alternative to the assumed natural one, are considered in the literature.
    Among all of them, we choose to consider the above size notion in order to
    minimise the technical overhead of the overall presentation.  Analytic
    methods employed in the current paper cover a broad range of possible size
    measures.  We refer the curious reader
    to~\cite{gryles2015,BodiniGitGol17,GittenbergerGolebiewskiG16} for a
    detailed analysis of various size notions in the de~Bruijn representation.
\end{rem}

Let $l_n$ denote the number of plain $`l$\=/terms of size $n$. Consider the
generating function $L_\infty(z) = \sum_{n \geq 0} l_n z^n$.  Using symbolic
methods, see~\cite[Part A.  Symbolic Methods]{flajolet09} we note that
$L_\infty(z)$ satisfies
\begin{equation}\label{eq:L_infty:fun:eq}
 L_\infty(z) \ = \ z L_\infty(z) + z {L_\infty(z)}^2 + D(z)
\qquad\textrm{where}\qquad
  D(z) = \frac{z}{1-z} = \sum_{n=0}^{\infty} z^{n+1}.
\end{equation}
In words, a $`l$\=/term is either (a) an abstraction followed by another
$`l$\=/term, accounting for the first summand, (b) an application of two
$`l$\=/terms, accounting for the second summand, or finally, (c) a de~Bruijn
index which is, in turn, a sequence of successors applied to $0$.
Solving~\eqref{eq:L_infty:fun:eq} for $L_\infty(z)$ we find that the generating
function $L_\infty(z)$, taking into account that the coefficients $l_n$ are
positive for all $n$, admits the following closed-form solution:
\begin{equation}
    L_{\infty}(z) = \frac{1-z-\sqrt{{(1-z)}^2 - \frac{4z}{1-z}}}{2z}\, .
\end{equation}
The first values of the coefficients of $L_{\infty}$ are:
\begin{center}
  1, 3, 10, 40, 181, 884, 4539, 24142, 131821,
 734577, 4160626 23881695,138610418, ...
\end{center}
This sequence is \textbf{A258973} in the \emph{Online Encyclopedia of Integer
  Sequences}.  In such a form, $L_\infty(z)$ is amenable to the standard techniques of
singularity analysis.
In consequence we have the following general asymptotic
approximation of $l_n$.
\begin{thmC}[Bendkowski, Grygiel, Lescanne, Zaionc~\cite{BendkowskiGLZ16}]
    The sequence $\left([z^n]L_\infty(z)\right)_n$ corresponding to plain
    $`l$\=/terms of size $n$ admits the following asymptotic approximation:
            \begin{equation}
            [z^n]L_{\infty}(z) \sim C {\rho^{-n}_{L_\infty}} n^{-3/2}
\end{equation}
where
\begin{equation}\label{eq:L_infty:sing}
\domsing{L_\infty} = \frac{1}{3}
        \left(\sqrt[3]{26+6 \sqrt{33}} -\frac{4\ 2^{2/3}}{\sqrt[3]{13+3
        \sqrt{33}}} - 1\right) \doteq 0.29559 \quad \text{and} \quad C \doteq 0.60676.
\end{equation}
\end{thmC}

In the context of evaluation, the arguably most interesting subclass of
$`l$\=/terms are closed or, more generally, $m$-open $`l$\=/terms. Recall that
an $m$-open $`l$\=/term takes one of the following forms. Either it is (a) an
abstraction followed by an $(m+1)$-open $`l$\=/term, or (b) an application of
two $m$-open $`l$\=/terms, or finally, (c) one of the indices
$\underline{0},\underline{1},\ldots,\underline{m-1}$.  Such a specification for
$m$-open $`l$\=/terms yields the following functional equation defining the
associated generating function
$L_m(z)$:
\begin{equation}\label{eq:def:L_m}
  L_m(z) = z L_{m+1}(z) + z {L_m(z)}^2 + \frac{1-z^m}{1-z}\, .
\end{equation}
Since $L_m(z)$ depends on $L_{m+1}(z)$,
solving~\eqref{eq:def:L_m} for $L_m(z)$ one finds that
\begin{equation}\label{eq:def:L_m:closed}
  L_m(z) = \frac{1-\sqrt{1 - 4z^2\left(L_{m+1}(z)
        + \frac{1-z^m}{1-z}\right)}}{2z}\, .
\end{equation}
For instance, the first coefficients of $L_0(z)$ are
\begin{equation*}
  0,0,1,1,3,6,17,41,116,313,895,2550,7450,21881,65168,\ldots
\end{equation*}
the first coefficients of $L_1(z)$ are
\begin{equation*}
 0,1,1,3,5,15,34,98,258,743,2098,6142,17988,53614,160619,\ldots
\end{equation*}
and the first coefficients of $L_2(z)$ are
\begin{equation*}
  0,1,2,3,8,18,49,130,364,1032,2987,8758,26000,77937,235677,\ldots
\end{equation*}

The presentation of $L_m(z)$ given in (\ref{eq:def:L_m}) poses considerable difficulties as $L_m(z)$
depends on $L_{m+1}(z)$ depending itself on $L_{m+2}(z)$, etc.  If developed,
the formula~\eqref{eq:def:L_m:closed} for $L_m(z)$ consists of an infinite
number of nested radicals. In consequence, standard analytic combinatorics tools
do not provide the asymptotic expansion of $[z^n]L_m(z)$, in particular
$[z^n]L_0(z)$ associated with closed $`l$\=/terms.  In their recent breakthrough
paper, Bodini, Gittenberger and Go\l\c{e}biewski~\cite{BodiniGitGol17} propose a
clever approximation of the infinite system associated with $L_m(z)$ and give
the following asymptotic approximation for the number of $m$-open $`l$\=/terms.
\begin{thmC}[Bodini, Gittenberger and
    Go\l\c{e}biewski~\cite{BodiniGitGol17}]\label{th:L_m:asymptotic:expansion}
    The sequence $\left([z^n]L_m(z)\right)_n$ corresponding to $m$-open
    $`l$\=/terms of size $n$ admits the following asymptotic approximation:
            \begin{equation}
            [z^n]L_{m}(z) \sim C_m {\rho^{-n}_{L_\infty}} n^{-3/2}
\end{equation}
where ${\rho_{L_\infty}}$ is the dominant singularity corresponding to plain
    $`l$-terms, see~\eqref{eq:L_infty:sing}, and $C_m$ is a constant, depending
    solely on $m$.
\end{thmC}
Let us remark that for closed $`l$\=/terms, the constant $C_0$ lies in between
$0.07790995266$ and $0.0779099823$.  In what follows, we use the
above~\autoref{th:L_m:asymptotic:expansion} in our investigations regarding what
we call closed closures.

% counting plain closures and environments
%%%% begin Figure "Recurrence of e_n"
\newcommand{\FigureRecurrence}{
\begin{figure}[ht!]
%    \begin{subfigure}[b]{0.6\textwidth}
  \scalebox{1.0}{\parbox{\linewidth}{%
  \begin{align*}
    & (125\,{n}^{3}-125\,n ) \;e_{n} \,+ \\
    &  (-475\,{n}^{3}-150\,{n}^{2}+325\,n ) \;e_{n+1}\,+ \\
    &  ( -1625\,{n}^{3}-13650\,{n}^{2}-29125\,n-17100 ) \;e_{n+2}\,+ \\
    &  ( 5925\,{n}^{3}+65550\,{n}^{2}+ 204825\,n+190800 ) \;e_{n+3}\,+ \\
    &  ( -10950\,{n}^{3}
    -149850\,{n}^{2}-609000\,n-744300 ) \;e_{n+4}\,+ \\
    &  ( 43599\,{n}^{3}+638460\,{n}^{2}+3028701\,n+4633680 ) \;e_{n+5}\,+ \\
    & ( -97781\,{n}^{3}-1680378\,{n}^{2}-9481237
    \,n-17550960 ) \;e_{n+6}\,+ \\
    &  ( 122749\,{n}^{3}+
    2388066\,{n}^{2}+15211685\,n+31648968 ) \;e_{n+7}\,+ \\
    &  ( -184402\,{n}^{3}-3954630\,{n}^{2}-27717140\,n-63149544
    ) \;e_{n+8}\,+ \\
    &  ( 280081\,{n}^{3}+6826380\,{n}^{2}+54868451\,n+145130568 ) \;e_{n+9}\,+ \\
    &  ( -205649 \,{n}^{3}-5654610\,{n}^{2}-51851989\,n-158722620 ) \;e_{n+10}\,+ \\
    &  ( 37439\,{n}^{3}+1339686\,{n}^{2}+16635271\,n+70682784
    ) \;e_{n+11}\,+ \\
    &  ( -68686\,{n}^{3}-3028038\,{n}^{2}-43616336\,n-205972920 ) \;e_{n+12}\,+ \\
    &  ( 222029 \,{n}^{3}+9258780\,{n}^{2}+128417911\,n+592399800 )
    \;e_{n+ 13}\,+ \\
    &  ( -241115\,{n}^{3}-10519830\,{n}^{2}-152823475\,n-
    739190880 ) \;e_{n+14}\,+ \\
    &  ( 134151\,{n}^{3}+
    6201222\,{n}^{2}+95476551\,n+489605640 ) \;e_{n+15}\,+ \\
    &  ( -42231\,{n}^{3}-2067834\,{n}^{2}-33729375\,n-183277332
    ) \;e_{n+16}\,+ \\
    &  ( 7470\,{n}^{3}+386418\,{n}^{2}+
    6659316\,n+38233296 ) \;e_{n+17}\,+ \\
    &  ( -678\,{n}^{3}-36972\,{n}^{2}-671670\,n-4065240 ) \;e_{n+18}\,+ \\
    &  ( 24\,{n}^{3}+1380\,{n}^{2}+26436\,n+168720 ) \;e_{n+ 19} = 0.
  \end{align*}
    }}
%    \end{subfigure}%

\medskip

%    \begin{subfigure}[b]{0.4\textwidth}
    \begin{math}
    \resizebox{0.55\textwidth}{!}{$\displaystyle
    \begin{array}{l@{\quad}l} %
      \begin{array}{ll}
        e_{ 0} &= 1,\\
        e_{ 1} &= 1,\\
        e_{ 2} &= 4,\\
        e_{ 3} &=17,\\
        e_{ 4} &=77,\\
        e_{ 5} &=364,\\
        e_{ 6} &=1776,\\
        e_{ 7} &=8881,\\
        e_{ 8} &=45296,\\
        e_{ 9} &=234806,
      \end{array}
      & %
      \begin{array}{ll}
        e_{ 10} &=1233816,\\
        e_{ 11} &=6558106,\\
        e_{ 12} &= 35202448,\\
        e_{ 13} &=190568779,\\
        e_{ 14} &= 1039296373,\\
        e_{ 15} &=5704834700,\\
        e_{ 16} &= 31494550253,\\
        e_{ 17} &=174759749005,\\
        e_{ 18} &= 974155147162.
      \end{array}
    \end{array}
    $}
    \end{math}
%    \end{subfigure}%

\caption{Linear recurrence defining $e_{n}$ with
corresponding initial conditions.}\label{fig:rec}
\end{figure}}
%%% End Figure

\section{Counting plain closures and environments}
\label{sec:counting-plain}
In this section we start with counting \emph{plain environments and closures},
i.e.~members of $\mathcal{E}nv$ and $\mathcal{C}los$,
see~\eqref{eq:closure:env:def}.  We consider a simple model in which the size of
environments and closures is equal to the total number of abstractions,
applications and the sum of all the de~Bruijn index sizes. Formally, we set
\begin{displaymath}
\left|"<" M, e ">"\right| = \left|M\right| + \left|e\right|\qquad\qquad
\left|\mathfrak{c} : e\right| = \left|\mathfrak{c}\right| + \left|e\right|\qquad\qquad
\left|\Box\right| = 0\, .
\end{displaymath}

\begin{exa}
    The following two tables list the first few plain environments and closures.
    \begin{displaymath}
    \resizebox{0.4\textwidth}{!}{$
    \begin{array}{l | c | c}
      \textbf{size} & \textbf{environments}& \textbf{total}\\\hline\hline
      0 & \Box & 1 \\\hline
      1 & "<" \zero, \Box ">":\Box & 1 \\\hline
      & "<"\zero, \Box">" :"<"\zero, \Box">" :\Box &\\
      2 & "<" \zero, "<"\zero, \Box">":\Box ">" : \Box & 4 \\
      & "<"`l \zero, \Box">":\Box,\quad
      "<" \one, \Box">" :\Box&
    \end{array}$}
    \qquad\qquad
    \resizebox{0.3\textwidth}{!}{$
    \begin{array}{l | c | c}
      \textbf{size} & \textbf{closures}& \textbf{total}\\\hline\hline
      0 & & 0 \\\hline
      1 & "<" \zero, \Box ">" & 1 \\\hline
      & "<" \zero, "<"\zero, \Box">" ">"  & \\
      2 & "<" `l\zero, \Box">" \quad "<" \one, \Box">"
      & 3
    \end{array}$}
  \end{displaymath}
\end{exa}
By analogy with the notation $\gL_\infty$ for the set of plain $`l$\=/terms, we
write $\E_\infty$ and $\C_\infty$ to denote the class of plain environments and
closures, respectively.  Reformulating~\eqref{eq:closure:env:def} we can now
give a formal specification for both $\E_\infty$ and $\C_\infty$ as follows:
\begin{align}\label{eq:plain:closures:envs:system}
\begin{split}
  \E_\infty &= \C_\infty : \E_\infty~\mid~\Box \\
  \C_\infty &= "<" \gL_\infty, \E_\infty ">"\, .
\end{split}
\end{align}
In such a form, both classes $\E_\infty$ and $\C_\infty$ become amenable to the
process of singularity analysis. In consequence, we obtain the following
asymptotic approximation for the number of plain environments and closures.

\begin{thm}
The numbers $e_n$ and $c_n$ of plain environments and closures of size $n$,
    respectively, admit the following asymptotic approximations:
\begin{equation}\label{eq:plain:closures:evns:approx}
e_n \sim C_e \cdot \rho^{-n} n^{-3/2} \quad \text{and} \quad
c_n \sim C_c \cdot \rho^{-n} n^{-3/2}
\end{equation}
where
\begin{align}
    \begin{split}
        C_e &= \frac{\sqrt{\frac{5}{47} \left(109+35
        \sqrt{545}\right)}}{8\,\sqrt{\pi}} \doteq 0.699997,\\
        C_c &=  \frac{\sqrt{\frac{10 \left(48069 \sqrt{5}-10295
        \sqrt{109}\right)}{65 \sqrt{109}-301
        \sqrt{5}}}}{\sqrt{\pi}\,\left(77-3\sqrt{545}\right)} \doteq
        0.174999
    \end{split}
\end{align}
and
\begin{equation}\label{eq:plain:environments:rho}
\rho = \dfrac{1}{10} \left(25-\sqrt{545}\right) \doteq 0.165476 \quad
\text{giving} \quad
\rho^{-n} \doteq {6.04315}^n.
\end{equation}
\end{thm}
\begin{proof}
Consider generating functions $E_\infty(z)$ and $C_\infty(z)$ associated with
    respective counting sequences, i.e.~the sequence ${\left(e_n\right)}_n$ of
    plain environments of size $n$ and ${\left(c_n\right)}_n$ of plain closures
    of size $n$.  Based on the
    specification~\eqref{eq:plain:closures:envs:system} for $\E_\infty$ and
    $\C_\infty$ and the assumed size notion, we can write down the following
    system of functional equations satisfied by $E_\infty(z)$ and $C_\infty(z)$:
\begin{align}\label{eq:plain:closures:envs:gfun}
\begin{split}
  E_{\infty}(z) &= C_{\infty}(z) E_{\infty}(z) + 1\\
  C_{\infty}(z) &= L_{\infty}(z)\, E_{\infty}(z).
\end{split}
\end{align}
Next, we solve~\eqref{eq:plain:closures:envs:gfun} for $E_{\infty}(z)$ and
    $C_{\infty}(z)$. Though~\eqref{eq:plain:closures:envs:gfun} has two formal
    solutions, the following one is the single one yielding analytic generating
    functions with non-negative coefficients:
\begin{equation}\label{eq:plain:closures:envs:gfun:solution}
E_{\infty}(z) = \frac{1-\sqrt{1-4 L_\infty(z)}}{2 L_\infty(z)}
\quad \text{and} \quad C_{\infty}(z) = \frac{1}{2} \left(1-\sqrt{1-4 L_\infty(z)}\right)\, .
\end{equation}
Since $L_\infty(z) > 0$ for $z \in \left(0, \rho_{L_\infty}\right)$ there are
    two potential sources of singularities
    in~\eqref{eq:plain:closures:envs:gfun:solution}.  Specifically, the
    dominating singularity $\rho_{L_\infty}$ of $L_\infty(z)$,
    see~\eqref{eq:L_infty:sing}, or roots of the radicand expression $1-4
    L_\infty(z)$.  Therefore, we have to determine whether we fall into the
    so-called sub- or super-critical composition schema, see~\cite[Chapter VI.
    9]{flajolet09}.  Solving $1 - 4 L_{\infty}(z) = 0$ for $z$, we find that it
    admits a single solution $\rho$ equal to
\begin{equation}
\rho = \frac{1}{10} \left(25-\sqrt{545}\right) \doteq 0.165476\, .
\end{equation}
    Since $\rho < \rho_{L_\infty}$ the outer radicand carries the dominant
    singularity $\rho$ of both $E_{\infty}(z)$ and $C_{\infty}(z)$. We fall
    therefore directly into the super-critical composition schema and in
    consequence know that near $\rho$ both $E_{\infty}(z)$ and $C_{\infty}(z)$
    admit Newton-Puiseux expansions in form of
\begin{align}\label{eq:plain:closures:envs:gfun:puiseux}
\begin{split}
E_{\infty}(z) &= a_{E_{\infty}} + b_{E_{\infty}} \sqrt{1- \frac{z}{\rho}} +
    O\left(\bigg|1-\frac{z}{\rho}\bigg|\right)\\
\text{and}\\
C_{\infty}(z) &= a_{C_{\infty}} + b_{C_{\infty}} \sqrt{1- \frac{z}{\rho}} +
    O\left(\bigg|1-\frac{z}{\rho}\bigg|\right)
\end{split}
\end{align}
with $a_{E_{\infty}}, a_{C_{\infty}} > 0$ and $ b_{E_{\infty}},
b_{C_{\infty}} < 0$. At this point, we can apply the standard function
scale, see~\autoref{th:standard-func-scale}, to the presentation of
$E_{\infty}(z)$ and $C_{\infty}(z)$ in
$(\ref{eq:plain:closures:envs:gfun:puiseux})$ and conclude that
\begin{equation}
[z^n]E_{\infty}(z) \sim C_{E_{\infty}} \rho^{-n} n^{-3/2}
\quad \text{and} \quad [z^n]C_{\infty}(z) \sim C_{C_{\infty}} \rho^{-n} n^{-3/2}
\end{equation}
where $C_{E_{\infty}} = \dfrac{b_{E_{\infty}}}{\Gamma(-\frac{1}{2})}$ and
    $C_{C_{\infty}} = \dfrac{b_{C_{\infty}}}{\Gamma(-\frac{1}{2})}$,
    respectively, with $\Gamma(-\frac{1}{2})= 2\sqrt{\pi}$.  In fact,
    reformulating~\eqref{eq:plain:closures:envs:gfun:solution} so to fit the
    Newton-Puiseux expansion forms~\eqref{eq:plain:closures:envs:gfun:puiseux}
    we find that
\begin{equation}
a_{E_{\infty}} = 2, \quad
b_{E_{\infty}} = -\frac{1}{4} \sqrt{\frac{5}{47} \left(109+35 \sqrt{545}\right)}
\end{equation}
and
\begin{equation}
a_{C_{\infty}} = \frac{1}{2}, \quad
b_{C_{\infty}} = \frac{2 \sqrt{\frac{10 \left(48069 \sqrt{5}-10295
    \sqrt{109}\right)}{65 \sqrt{109}-301 \sqrt{5}}}}{3 \sqrt{545}-77}
\end{equation}
Numerical approximations of $C_{E_{\infty}} =
\dfrac{b_{E_{\infty}}}{\Gamma(-\frac{1}{2})}$ and $C_{C_{\infty}} =
\dfrac{b_{C_{\infty}}}{\Gamma(-\frac{1}{2})}$ yield the declared asymptotic
behaviour of ${\left(e_n\right)}_n$ and ${\left(c_n\right)}_n$,
see~\eqref{eq:plain:closures:evns:approx}.
\end{proof}

Let us notice that as both generating functions $E_{\infty}(z)$ and
$C_{\infty}(z)$ are algebraic, they are also \emph{holonomic} (D-finite),
i.e.~satisfy differential equations with polynomial (in terms of~$z$)
coefficients. Using the powerful \texttt{gfun} library for
\texttt{Maple}~\cite{SalvyZimmermann1994} one can automatically derive
appropriate holonomic equations for $E_{\infty}(z)$ and $C_{\infty}(z)$,
subsequently converting them into linear recurrences for sequences
${\left(e_n\right)}_n$ and ${\left(c_n\right)}_n$.

\begin{exa}\label{exa:linEn}
We restrict the presentation to the linear recurrence for the number of plain
    environments, omitting for brevity the, likely verbose, respective
    recurrence for plain closures.  Using \texttt{gfun} we find that $e_{n}$
    satisfies the recurrence of Figure~\ref{fig:rec}.  Despite its appearance,
    this recurrence is an efficient way of computing $e_{n}$.  Indeed, holonomic
    specifications for $C_{\infty}(z)$ and $E_{\infty}(z)$ allow  computing the
    coefficients $[z^n]C_{\infty}(z)$ and $[z^n]E_{\infty}(z)$ using a linear
    number of arithmetic operations, as opposed to a quadratic number of
    operations as following their direct combinatorial specification. Let us
    remark that the involved computations operate on large integers, which
    have a linear in $n$
    space representation. For instance, $e_{1000}$ has about $600$
    digits.  In consequence, single arithmetic operations on such numbers cannot
    be performed in constant time.
\FigureRecurrence
\end{exa}

\section{Random generation of closures and environments}
\label{sec:random}

Effective counting methods for various discrete structures are among the most
prominent and ubiquitous subjects in combinatorics. Although interesting in
their own right, such counting methods (and in particular related algorithms)
exhibit important benefits in the context of generating random instances of
corresponding combinatorial structures. Let us mention, for instance, the
successive use of random $`l$\nobreakdash-terms used to disprove the correctness
of eagerness optimisations of the salient Glasgow Haskell Compiler, see~\cite{palka2012}.

Given the fact that closures and environments are fundamental data structures
used in different abstract machines related to the execution of
$`l$\nobreakdash-terms, random closures and environments can be used to model
(in other words simulate) actual data encountered in the execution traces of
abstract machines such as the Krivine or U-machines.  In this context, random
generation of closures and environments provide effective means of testing the
correctness of respective abstract machine implementations as well as facilitate
their optimisation and eventual perfection.

With analytic generating functions $C_\infty(z)$ and $E_\infty(z)$ for plain
closures and environments, respectively, it becomes possible to design efficient
exact- or approximate-size samplers (i.e.~algorithms constructing random
structures) corresponding to both combinatorial classes. In particular, we can
use the general frameworks of Boltzmann samplers~\cite{Duchon2004} by Duchon et
al.~or the so-called recursive method~\cite{NijenhuisWilf1978,FLAJOLET19941} of
Nijenhuis and Wilf. Remarkably, in both frameworks the sampler design resembles
the recursive structure of the target combinatorial specification.  Moreover,
for a broad class of discrete structures such as, for instance, algebraic
specifications, the sampler construction itself can be effectively automatised.
Respective branching probabilities dictating the sampler's decisions are
precomputed once and fixed throughout all subsequent executions.

In the recursive method, branching probabilities are established so to obtain an
exact-size sampler, i.e.~a sampler which generates random structures of a
specific, given size $n$. In particular, using holonomic specifications for
$C_\infty(z)$ and $E_\infty(z)$ it is possible to compute the related
coefficients $[z^n]C_\infty(z)$ and $[z^n]E_\infty(z)$ using just $O(n)$
arithmetic operations, thus reach larger target sizes in a reasonable amount of
time.  On the other hand, if we drop the exact-size requirement of the outcome
structures, it is possible to (again, automatically) construct an
approximate-size Boltzmann sampler generating closures (respectively environments) of
varying size in linear time, in terms of outcome size. Although the output size
of constructed objects is itself random, it is possible to \emph{calibrate}
its expectation around a (not necessarily) finite mean. Furthermore, using an optional
rejection phase, meant to dismiss inadmissible structures, it is possible to
gain additional control over the  sampler outcome.

Remarkably, both mentioned sampler frameworks admit effective tuning procedures
influencing the expected internal shape of constructed objects, e.g.~frequencies
of desired sub-patterns~\cite{DBLP:conf/analco/BendkowskiBD18}. It is therefore
possible to control the expected internal structure of the generated closures and
environments.

We offer prototype sampler implementations for plain environments and closures,
within the above sampler frameworks at
Github\footnote{\url{https://github.com/PierreLescanne/CountingAndGeneratingClosuresAndEnvironments}}.
Likewise, we provide similar samplers for so-called closed closures and
environments (see Section~\ref{sec:counting-closed}) based on the
recursive method.

% counting closed closures and environments
%\input{sections/5-counting-closed-closures-and-environments}
\section{Counting closed closures}
\label{sec:counting-closed}
A closure $"<" M, e ">"$ is said to be \emph{$m$\nobreakdash-open}, denoted also
as $"<" M, e ">"`:\Clos_m$, if there exists a non-negative $p$ such that
$M`:\gL_{m+p}$ (i.e.~$M$ is an $(m+p)$-open $`l$-term) and $e$ is a finite list
(i.e.~environment) of length $p$ consisting itself of $m$-open closures.  In
other words, $m$-open closures are structures defined by means of the following
implicit combinatorial specification:
\begin{equation}\label{eq:closed:closures:envs:system}
  \Clos_m ::= \gL_m \times \Box ~\mid~
   \gL_{m+1} \times  "<" \Clos_m ">" ~\mid~
   \gL_{m+2} \times  "<" \Clos_m, \Clos_m ">" ~\mid~\cdots
\end{equation}
In particular, a closure is said to be \emph{closed}\footnote{We acknowledge
that speaking of closed closures is a bit odd, however terms ``closure'' and
``closed'' form a consecrated terminology that we merely associate together.} if
it is $0$-open. Like in the case of $m$-open $`l$-terms, if a closure $"<" M, e
">"$ is $m-$open, then it is also $(m+1)$-open.
Consider the following example:

\begin{exa}~
  \begin{itemize}
  \item $"<"\zero\,\one,"<"`l\underline{3}">":\Box">"$ is a $3$-open
    closure.
  \item $"<"\one\, \zero, "<"`l\zero,\Box">":"<"`l`l\zero,\Box">":\Box">"$
    is a closed closure ($0$-open closure).
  \end{itemize}
\end{exa}
Let us remark that an $m$\nobreakdash-open closure corresponds to a not yet
evaluated $m$\nobreakdash-open $`l$\nobreakdash-term. Certainly, due to their
ubiquity in the context of abstract machines, the most interesting
$m$-open closures are in fact closed.  In the current section, we focus
therefore on counting closed closures and corresponding closed environments.

\begin{exa} The following table lists the
  first few closed closures.
  \begin{displaymath}
    \resizebox{0.6\textwidth}{!}{$
    \begin{array}{l | c | c}
      \textbf{size} & \textbf{closures}& \textbf{total}\\\hline\hline
0,1 & &0\\\hline
      2  & "<"`l\zero , \Box">" & 1 \\\hline
      3 & "<"`l`l\zero , \Box">" \quad "<"\zero ,  "<"`l\zero , \Box">"">" & 2 \\\hline
      &  "<"`l`l`l\zero , \Box">" \quad "<"`l`l\one , \Box">" \quad "<"`l(\zero \zero) , \Box">" & \\
      4 &     "<"`l\zero ,  "<"`l\zero , \Box">"">" \quad "<"\zero ,  "<"`l`l\zero , \Box">"">" \quad "<"\zero ,  "<"\zero ,  "<"`l\zero , \Box">"">"">"     & 6
    \end{array}$}
  \end{displaymath}
\end{exa}

\autoref{fig:nbCloClo} gives the first $50$ numbers of closed closures.

\begin{figure}[!ht]
  \begin{displaymath}
    \begin{array}[t]{r|r||r|r}
      \textbf{n}& c_{0,n} &\textbf{n}& c_{0,n}\\\hline
      0& 0& 25 & 2039291268600\\
      1& 0& 26 &7690787869550\\
      2& 1& 27  &  29071665271653\\
      3& 2& 28  &  110130490287410\\
      4& 6& 29  &  418043342219865\\
      5& 18& 30 &  1589843149170521\\
      6& 58& 31  &  6056959298323505\\
      7& 188& 32 &  23113998858734867\\
      8& 630& 33  &  88343015816573484\\
      9& 2140& 34 &  338147576768474959\\
      10& 7384& 35 &  1296106542004047500\\
      11& 25775& 36  &  4974412840517200748\\
      12& 90919& 37  &  19115189068830345885\\
      13& 323529& 38 &  73539781161982872915\\
      14& 1160285& 39 &  283234718823200209560\\
      15& 4189666& 40  &  1092009621308203935814\\
      16& 15221235& 41 &  4214435736178031843666\\
      17& 55602475& 42 &  16280366813995192858378\\
      18& 204119165& 43 &  62947860010954764058213\\
      19& 752691547& 44  &  243596693995304845906020\\
      20& 2786900678& 45 &  943448667650667612945764\\
      21& 10357265495& 46 &  3656836859592859541767133\\
      21& 38623769249& 47  &  14184639891328996401070032\\
      23& 144488013135& 48 &  55060786067960705278258741\\
      24& 542090016461 & 49 & 213877295469617703331719718 \\
    \end{array}
  \end{displaymath}
  \caption{The number of closed closures for $n=0,\ldots,49$}
  \label{fig:nbCloClo}
\end{figure}
Establishing the asymptotic growth rate of the sequence
${\left(c_{0,n}\right)}_n$ corresponding to closed closures of size $n$ poses a
considerable challenge, much more involved than its plain counterpart.  In the
following theorem we show that there exists two constants $\underline{\rho},
\overline{\rho} < \rho_{L_\infty}$ such that $\displaystyle\lim_{n \to \infty}
\dfrac{{\underline{\rho}}^{-n}}{c_{0,n}} = 0$ and $\displaystyle\lim_{n \to
\infty} \dfrac{c_{0,n}}{{\overline{\rho}}^{-n}} = 0$.  In other words, the
asymptotic growth rate of ${\left(c_{0,n}\right)}_n$ is bounded by two
exponential functions of $n$.

\begin{thm}
There exist $\overline{\rho} < \underline{\rho}$ satisfying
$\overline{\rho} < \underline{\rho} < \rho_{L_{\infty}}$ and
    functions $\theta(n), \kappa(n)$ satisfying $\displaystyle \limsup_{n \to
    \infty} {\theta(n)}^{1/n} = \limsup_{n \to \infty} {\kappa(n)}^{1/n} = 1$
    such that for sufficiently large $n$ we have
    ${\underline{\rho}}^{-n}\theta(n) < c_{0,n} <
    {\overline{\rho}}^{-n}\kappa(n)$.
\end{thm}
\begin{proof}
Let us start with the generating function $C_0(z)$ associated with closed
    closures $\Clos_0$. Note that from the
    specification~\eqref{eq:closed:closures:envs:system}, instantiated to $m=0$,
    $C_0(z)$ is implicitly defined as
\begin{equation}\label{eq:C_0(z):impl:def}
C_0(z) = \sum_{m \geq 0} L_m(z) {C_0(z)}^m .
\end{equation}
We can therefore identify a closed closure $\mathfrak{c}$ with a tuple
    $(t,c_1,\ldots,c_m)$ where $m \geq 0$, $t$ is an $m$-open $`l$-term and
    $c_1,\ldots,c_m$ are closed closures themselves.  We proceed with defining
    two auxiliary lower and upper bound classes $\underline{C}_0(z)$ and
    $\overline{C}_0(z)$ such that $ [z^n]\underline{C}_0(z) \leq [z^n] C_0(z)
    \leq [z^n]\overline{C}_0(z)$ for all $n$. Next, we establish their asymptotic
    behaviour and, in doing so, provide exponential lower and upper bounds on
    the growth rate of closed closures.

We start with $\underline{C}_0(z) = \sum_{m \geq 0} L_0(z)
    {\underline{C}_0(z)}^m$. Note that $\underline{C}_0(z)$ is associated with
    closures in which each term is closed, independently of the corresponding
    environment length. Hence, as closed $`l$\=/terms are $m$-open for all $m
    \geq 0$, we have $[z^n]\underline{C}_0(z) \leq [z^n]C_0(z)$.  Furthermore
\begin{equation}
\underline{C}_0(z) = \sum_{m \geq 0} L_0(z) {\underline{C}_0(z)}^m
=  L_0(z) \sum_{m \geq 0} {\underline{C}_0(z)}^m
=   \frac{L_0(z)}{1 - \underline{C}_0(z)}.
\end{equation}
Solving the above equation for $\underline{C}_0(z)$ we find that
    $\underline{C}_0(z) = \frac{1}{2} \left(1-\sqrt{1-4 L_0(z)}\right)$.  In
    such a form, it is clear that there are two potential sources of
    singularities, i.e.~the singularity $\rho_{L_\infty}$ of $L_0(z)$,
    see~\autoref{th:L_m:asymptotic:expansion}, or the roots of the radicand $1-4
    L_0(z)$.  Since $L_0(z)$ is increasing and continuous in the interval
    $(0,\rho_{L_\infty})$ we know that if $L_0(\rho_{L_\infty}) > \frac{1}{4}$
    then there exists a $\underline{\rho} < \rho_{L_\infty}$ such that
    $L_0(\underline{\rho}) = \frac{1}{4}$. Unfortunately, we cannot simply check
    that $L_0(\rho_{L_\infty}) > \frac{1}{4}$ as there exists no known method of
    evaluating $L_0(z)$, defined by means of an infinite system of equations, at
    a given point. For that reason we propose the following approach.

Recall that a $`l$\=/term $M$ is said to be \emph{$h$-shallow} if all its
    de~Bruijn index values are (strictly) bounded by $h$,
    see~\cite{GittenbergerGolebiewskiG16}. Let $L^{(h)}_m(z)$ denote the
    generating function associated with $m$-open $h$-shallow $`l$-terms.  Note
    that $L^{(h)}_0(z)$, i.e.~the generating function corresponding to closed
    $h$-shallow $`l$-terms, has a finite computable representation. Indeed, we
    have
\begin{align}
\begin{split}
L^{(h)}_0(z) &= z L^{(h)}_1(z) + z L^{(h)}_0(z) L^{(h)}_0(z)\\
L^{(h)}_1(z) &= z L^{(h)}_2(z) + z L^{(h)}_1(z) L^{(h)}_1(z) + z\\
L^{(h)}_2(z) &= z L^{(h)}_3(z) + z L^{(h)}_2(z) L^{(h)}_2(z) + z + z^2\\
\ldots\\
L^{(h)}_{h-1}(z) &= z L^{(h)}_h(z) + z L^{(h)}_{h-1}(z) L^{(h)}_{h-1}(z) + z + z^2 + \cdots + z^{h-1}\\
L^{(h)}_{h}(z) &= z L^{(h)}_h(z) + z L^{(h)}_{h}(z) L^{(h)}_{h}(z) + z + z^2 + \cdots + z^h
\end{split}
\end{align}
Consider $m < h$. Each $m$-open $h$-shallow $`l$-term is either (a) in form of
    $`l M$ where $M$ is an $(m+1)$-open $h$-shallow $`l$-term due to the head
    abstraction, (b) in form of $M N$ where both $M$ and $N$ are $m$-open
    $h$-shallow $`l$-terms, or (c) a de~Bruijn index in the set
    $\{\underline{0},\underline{1},\ldots,\underline{m-1}\}$. When $m = h$, we
    have the same specification with the exception of the first summand $z
    L^{(h)}_h(z)$ where, as we cannot exceed $h$, terms under abstractions are
    $h$-open, instead of $(h+1)$-open.

Using such a form it is possible to evaluate $L^{(h)}_0(z)$ at each point $z \in
(0,\rho_{(h)})$ where $\rho_{(h)} > \rho_{L_\infty}$ is the dominating
singularity of $L^{(h)}_0(z)$ satisfying $\rho_{(h)} \xrightarrow[h \to
\infty]{} \rho$, see~\cite{GittenbergerGolebiewskiG16}.  Certainly, each closed
$h$-shallow $`l$\=/term is in particular a closed $`l$\=/term.  In consequence,
$[z^n]L^{(h)}_0(z) \leq [z^n]L_0(z)$ for each $n$.  Moreover, for all
sufficiently large $n$ we have $[z^n]L^{(h)}_0(z) < [z^n]L_0(z)$.  This
coefficient-wise lower bound transfers onto the level of generating function
values and we obtain $L^{(h)}_0(z) < L_0(z)$. Following the same argument, we
also have $L^{(h)}_0(z) < L^{(h+1)}_0(z)$ for each $h \geq 1$. We can therefore
use $L^{(h)}_0(z)$ to approximate $L_0(z)$ from below --- the higher $h$ we
choose, the better approximation we obtain.  Using computer algebra
software\footnote{\url{https://github.com/PierreLescanne/CountingAndGeneratingClosuresAndEnvironments}}
it is possible to automatise the evaluation process of
$L^{(h)}_0(\rho_{L_\infty})$ for increasing values of $h$ and find that for $h =
153$ we obtain
\begin{equation}
    L_0^{(153)}(\rho_{L_\infty}) \doteq 0.25000324068941554\, .
\end{equation}
Hence indeed, the asserted existence of $\underline{\rho} <
\rho_{L_\infty}$ such that $L_0(\underline{\rho}) = \frac{1}{4}$ follows
(interestingly, taking $h = 152$ does not suffice as
$L_0^{152}(\rho_{L_\infty}) < \frac{1}{4}$).  We fall hence in the
super-critical composition schema\footnote{\emph{Supercriticality} ensures
  that meromorphic asymptotics applies and entails strong statistical
  regularities (see~\cite{flajolet09} Section V.2 and Section IX.6).} and note that $\underline{C}_0(z)$
admits a Newton-Puiseux expansion near $\underline{\rho}$ as follows:
\begin{equation}
\underline{C}_0(z) = \underline{a_0} -
    \underline{b_0}\sqrt{1-\frac{z}{\underline{\rho}}} + O\left(\bigg|
    1-\frac{z}{\underline{\rho}}\bigg|\right)
\end{equation}
for some constants $\underline{a_0} > 0$ and $\underline{b_0} < 0$.  Hence,
$[z^n] C_0(z)$ grows asymptotically faster than ${\underline{\rho}}^{-n}
\theta(n)$ where $\theta(n) = \dfrac{\underline{b_0}}{\Gamma(-\frac{1}{2})}
n^{-3/2}$.

For the upper bound we consider $\overline{C}_0(z) = \sum_{m \geq 0} L_\infty(z)
{\overline{C}_0(z)}^m$, i.e.~the generating function associated with closures
in which all terms are plain (either closed or open), independently of the
constraint imposed by the corresponding environment length. Following the
same arguments as before, we note that $[z^n]\overline{C}_0(z) > [z^n]C_0(z)$.
Now
\begin{equation}
\overline{C}_0(z) = \sum_{m \geq 0} L_\infty(z) {\overline{C}_0(z)}^m
=  L_\infty(z) \sum_{m \geq 0} {\overline{C}_0(z)}^m
=   \frac{L_\infty(z)}{1 - \overline{C}_0(z)}.
\end{equation}
Solving the equation for $\overline{C}_0(z)$ we find that $\overline{C}_0(z) =
\frac{1}{2} \left(1-\sqrt{1-4 L_\infty(z)}\right)$.  Note that in this case, we
can easily handle the radicand expression $1-4 L_\infty(z)$ and find out that,
as in the lower bound case, we are in the super-critical composition schema.
Specifically, $\overline{\rho} = \frac{1}{10} \left(25-\sqrt{545}\right) \doteq
0.165476$, cf.~\eqref{eq:plain:environments:rho}, is the dominating singularity
of $\overline{C}_0(z)$. In consequence, $\overline{C}_0(z)$ admits the following
Newton-Puiseux expansion near $\overline{\rho}$: \begin{equation}
    \overline{C}_0(z) = \overline{a_0} -
    \overline{b_0}\sqrt{1-\frac{z}{\overline{\rho}}} + O\left(\bigg|
    1-\frac{z}{\overline{\rho}}\bigg|\right) \end{equation} for some constants
$\overline{a_0} > 0$ and $\overline{b_0} < 0$.  In conclusion, $[z^n] C_0(z)$
grows asymptotically slower than ${(\overline{\rho})}^{-n} \theta(n)$
where $\theta(n) = \dfrac{\overline{b_0}}{\Gamma(-\frac{1}{2})}
n^{-3/2}$, finishing the proof.
\end{proof}
With an implicit expression defining $C_0(z)$, see~\eqref{eq:C_0(z):impl:def},
efficient random generation of closed closures poses a difficult task. Though we
have no efficient Boltzmann samplers, it is possible to follow the recursive
method and obtain exact-size samplers for a moderate range of target sizes.
We offer a prototype sampler of this kind, available at
Github\footnote{\url{https://github.com/PierreLescanne/CountingAndGeneratingClosuresAndEnvironments}}.

% conclusions
\section{Conclusions}
\label{sec:conclusions}
We view our contribution as a small step towards the quantitative,
average-case analysis of evaluation complexity in
$`l$\=/calculus. Using standard tools from analytic combinatorics, we
investigated some combinatorial aspects of environments and closures
--- fundamental structures present in various formalisms dealing with
normalisation in $`l$\=/calculus, especially in its variants with
explicit substitutions~\cite{Lescanne1994,DBLP:conf/ppdp/BendkowskiL18}. Though plain environments
and closures are relatively easy to count and generate, their closed
counterparts pose a considerable combinatorial challenge. The implicit
and infinite specification of closed closures based on closed
$`l$\=/terms complicates significantly the quantitative analysis,
namely estimating the exponential factor in the asymptotic growth
rate, or effectively generating random closed closures. In particular,
getting more parameters of the asymptotic growth will require more
sophisticated methods, like, for instance, the recent infinite system
approximation techniques of Bodini, Gittenberger and
Go\l\c{e}biewski~\cite{BodiniGitGol17}.

\bibliographystyle{plainurl}
\bibliography{references}

\end{document}